\documentclass[11pt,fleqn,a4paper]{article}

\usepackage{ucs}

\usepackage{graphicx}

\usepackage{amsmath}
\usepackage{amsfonts}
\usepackage{amssymb}
\usepackage{amsthm}

\usepackage{float}
\usepackage{multicol}

\usepackage{color}

\usepackage[nice]{units}

\newcommand{\R}{\mathbb{R}}

\newcommand{\C}{\mathbb{C}}

\newcommand{\Z}{\mathbb{Z}}

\newcommand{\Q}{\mathbb{Q}}

\newcommand{\rank}{\operatorname{rank}}

\newcommand{\imag}{\operatorname{i}}

\newcommand{\Abs}[1]{\left| #1 \right|}

\newcommand{\pfrac}[2]{\frac{\partial #1}{\partial #2}}

\newcommand{\noicefrac}[2]{#1/#2}

\newtheorem{Theorem}{Theorem}
\newtheorem{Lemma}[Theorem]{Lemma}
\newtheorem{Proposition}[Theorem]{Proposition}
\newtheorem{Corollary}[Theorem]{Corollary}
\newtheorem{Definition}[Theorem]{Definition}
\newtheorem{Remark}[Theorem]{Remark}

\newenvironment{remark}{\begin{Remark} \begin{rm}}{\end{rm} \end{Remark}}

\newenvironment{theorem}{\begin{Theorem} \begin{sl}}{\end{sl}
                         \end{Theorem}}
\newenvironment{lemma}{\begin{Lemma} \begin{sl}}{\end{sl} \end{Lemma}}
\newenvironment{proposition}{\begin{Proposition}
        \begin{sl}}{\end{sl} \end{Proposition}}

\title{Dynamics near the $p:-q$ Resonance}
\author{Sven Schmidt$^{1,3}$ and Holger R. Dullin$^{4}$\\ \\
\parbox{8cm}{
\begin{center}
$^{1}$School of Mathematics\\Loughborough University\\LE11 3TU, Loughborough, UK \\
\end{center}
\begin{center}
$^{3}$Schlumberger Oilfield UK Plc\\Abingdon Technology Centre\\OX14 1UJ, Abingdon, UK \\
{\tt sschmidt@abingdon.oilfield.slb.com}
\end{center}
}
\hfill
\parbox{8cm}{
\begin{center}
$^{4}$School of Mathematics and Statistics\\The University of Sydney\\Sydney NSW 2006, Australia
{\tt hdullin@usyd.edu.au}
\end{center}
}
}

\date{\today}

\newcommand{\cred}{}

\begin{document}

\maketitle

\begin{abstract}
We study the dynamics near the truncated $p:\pm q$ resonant Hamiltonian equilibrium for $p$, $q$ coprime.
The critical values of {\cred{}the} momentum map of the Liouville integrable system are found.
The three basic objects reduced period, rotation number, and non-trivial action for the leading order dynamics
are computed in terms of complete hyperelliptic integrals. A relation between the three functions that can
be interpreted as a decomposition of the rotation number into geometric and {\cred{}dynamic} phase is found.
Using this relation we show that the $p:-q$ resonance has fractional monodromy. Finally we prove that near the
origin of the $1:-q$ resonance the twist vanishes.

\vspace{0.5cm}

\vspace{0.5cm}
\end{abstract}

\section{Introduction}

The phase space of an integrable system is foliated into invariant tori almost everywhere. Exceptions occur when the integrals are not independent, i.e.~when the energy-momentum map has critical values. Monodromy is a manifestation of the singular nature of the fibre over a critical value of the energy--momentum mapping. The term was introduced in Duistermaat's 1980 paper~\cite{Duistermaat1980} as the simplest obstruction to the existence of global and smooth action--angle coordinates. Consider a closed, non--contractible loop $\Gamma$ in the image of the energy--momentum mapping consisting entirely of regular values. Assume the loop is traversed from $\Gamma(0)$ to $\Gamma(1)$ where $\Gamma(0) = \Gamma(1)$. For each regular value $\Gamma(s)$ the period lattice $P_{\Gamma(s)}$ gives the periods of the flows generated by energy and momentum, and thus encodes the transformation to action-angle variables. The period lattice $P_{\Gamma(s)}$ at $\Gamma(0)$ and at $\Gamma(1)$ are related by a unimodular transformation. If this transformation is non-trivial the system is said to have monodromy. Cushman and Bates~\cite{CushmanBates1997} use the ratios of periods, i.e.~the rotation number, to compute monodromy.
Geometrically monodromy means that the torus bundle over the loop $\Gamma$ is non-trivial.

If the loop of regular values is contractible through regular values there is no monodromy. Thus there must be  {\cred{}at least one} critical {\cred{}value} inside the loop {\cred{}$\Gamma$ for monodromy to occur}, and the nature of the critical fibre in principle determines the monodromy. For non--degenerate focus--focus points the critical fibre is a pinched torus with $p$ pinches, and the monodromy matrix of a loop around the critical value of the focus-focus point is $\begin{pmatrix} 1 & 0 \\ p & 1\end{pmatrix}$ \cite{Matveev96,NTienZung1997}.


Recently Nekhoroshev~et~al.~\cite{NekhoroshevSadovskiiZhilinskii2002,NekhoroshevSadovskiiZhilinskii2006} found fractional monodromy in certain resonant oscillators. In fractional monodromy the change of basis of the period lattice is not {\cred{}an element of} $SL(2, \Z)$ but {\cred{}of} $SL(2, \Q)$ {\cred{}instead}. This would be incompatible with the Liouville-Arnold theorem which gives the actions uniquely up to transformations {\cred{}in} $SL(2, \Z)$. Thus for fractional monodromy the loop $\Gamma$ necessarily has to cross critical values. The critical fibre has to be such that it is possible to continuously pass a sub-lattice of the period lattice through it \cite{NekhoroshevSadovskiiZhilinskii2006,Efstathiou2006}, also see \cite{Efstathiou2005}. The first example was the $1:-2$ resonant Hamiltonian, with an appropriate compactification \cite{NekhoroshevSadovskiiZhilinskii2002,NekhoroshevSadovskiiZhilinskii2006,Efstathiou2006}. In \cite{Nekhoroshev07} and \cite{SMPJJ08} this has been extended to the compactified $p:-q$ resonance.

The present paper is based on the thesis~\cite{Schmidt08}.
It is a generalization from the $1:-2$ resonance as treated in \cite{RHCushmanHRDullinHHSS2007}
to the general $p:-q$ resonance.
We independently obtained results similar to those presented in~\cite{SMPJJ08}, but our approach is complementary.
There are two main differences.
The first main difference is that we base our approach on the computation of the action variables. We derive the remarkable formula relating the action variables $I_1$ and $I_2$ by
\begin{equation} \label{geophas}
  I_2 = \frac{p+q}{4 \pi} \Delta h T - I_1 W
\end{equation}
where $\Delta h$ is {\cred{}the} value of the non-linear part of the Hamiltonian,
$T$ is the period,
and $W$ is the corresponding rotation number.
This formula can be interpreted as decomposing the rotation number into a dynamical phase proportional to $T$ and a geometric phase proportional to the action $I_2$ \cite{Montgomery91}.
Then we show that $W$ changes by $1/(pq)$ upon traversing a loop around the equilibrium,
and thus by \eqref{geophas} the action $I_2$ changes by $I_1/(pq)$ and the system has fractional monodromy.
The second main difference is that we analyze the dynamics near the equilibrium point in detail. The period, rotation number, and action are dominated by certain singular contributions, and we show that sufficiently close to the equilibrium point these are the leading order contributions. This allows us to perform the computation without any compactification, thereby proving the conjecture made in \cite{NekhoroshevSadovskiiZhilinskii2006}, that
fractional monodromy is independent of the details of the compactification.
Finally we obtain the new result that the isoenergetic non-degeneracy condition of the KAM theorem is violated
near the $1:-q$ resonance, i.e.~the system has vanishing twist near the equilibrium point.


The paper is structured as follows. First we derive the resonant Hamiltonian normal form near an elliptic--elliptic equilibrium point in $p:-q$ resonance which is then studied in the reduced phase space after applying singular reduction. After discussing the structure of the reduced phase space we derive the set of critical values of the energy--momentum mapping. The period of the reduced flow, the rotation number, and the non-trivial action of the system are computed in the next section.
We then prove fractional monodromy for the $p:-q$ resonance.
Finally we prove that the twist vanishes in $1:-q$ resonant systems close to the equilibrium point.

\section{Resonant Hamiltonain Equilibria}
\label{example}

Near a non-resonant elliptic equilibrium point a Hamiltonian can be transformed to Birkhoff normal form
up to arbitrary high order. The truncated Birkhoff normal form depends on the actions only, and to lowest
order these actions are those of the two harmonic oscillators of the diagonalized quadratic part of the Hamiltonian.
When the frequencies are $p: \pm q$ resonant, the normal form contains additional so called resonant terms that
depend on resonant linear combinations of the angles. Even though the non-linear non-resonant terms
are generically present, we will assume that all the non-linear non-resonant {\cred{}terms} up to the order $p+q$ are vanishing
Without this assumption the dynamics near the equilibrium would be dominated by the non-resonant
terms when $p+q\ge 5$. For the low order resonant cases no such assumption is necessary, and
therefore the results presented in  \cite{RHCushmanHRDullinHHSS2007} for the $1:-2$ resonance are
relevant for generic Hamiltonian systems. By contrast, the integrable systems studied in this paper have higher
codimension (increasing with $p+q$) in the class of Hamiltonian systems.

Denote a point in phase space $ T^{*}\R^{2}$ with coordinates  $( p_{1}, q_{1}, p_{2}, q_{2} ) $
and let the symplectic structure be $\omega = dq_{1} \wedge dp_{1} + dq_{2} \wedge dp_{2}$.
Assume the origin $(p_{1},q_{1},p_{2},q_{2})=(0,0,0,0) \in T^{*}\R^{2}$ is an elliptic equilibrium point of the system whose eigenvalues $\pm \imag \omega_{1}$, $\pm \imag \omega_{2}$ have ratio $p / q$, where $p$ and $q$ coprime positive integers. Then the quadratic part of the Hamiltonian $H$ near the equilibrium point can be brought into the form
\begin{equation}
	H_{2} = \frac{p}{2} \left( p_{1}^{2} + q_{1}^{2} \right)
	+ \sigma \frac{q}{2} \left( p_{2}^{2} + q_{2}^{2} \right) \,,
\end{equation}
by scaling $\omega_1$ to $p$ by a linear change of time.
For $\sigma = +1$ the quadratic Hamiltonian $H_2$ is definite, while for $\sigma = -1$ it is indefinite.
The system is said to be in $p : \pm q $ resonance, and mostly we are interested in the case $\sigma = -1$.

The classical treatment of resonant equilibria, see e.g.~\cite{VIArnoldVVKozlovAINeishtadt1997}, is as follows.
Denote $(A_{i}, \phi_{i})$, $i=1,2$, the canonical polar coordinates corresponding to $(p_{i}, q_{i})$. The resonant Birkhoff normal form
depends on $A_1$, $A_2$ and the resonant combination (the so called secular term) $-\sigma q\phi_1 + p \phi_2$.
The normal form Hamiltonian truncated at order $p+q$ is
\[
   H = p A_1 + \sigma q A_2 + \sum \mu_{ij} A_1^i A_2^j +  \mu \sqrt{A_1^q A_2^p} \cos (  - \sigma q\phi_1 + p \phi_2 + \varphi ) \,,
\]
where the phase $\varphi$ can be set to zero by a shift of the angles.
In order to reduce the number of variables a linear symplectic transformation to $(J_i, {\cred{}\psi_i})$ is performed.
The resonant combination is introduced as a new angle
${\cred{}\psi_2}$, while the cyclic angle ${\cred{}\psi_1}$ is conjugate to  $J_{1} = p A_{1} + \sigma q A_{2} = H_{2}$.
The complete transformation is given by $J = M^{-t} A$, ${\cred{}\psi} = M \phi$,
where $M$ contains two arbitrary integers $a$ and $b$ restricted by $\det M =  b p  -  \sigma a q = 1$.
The new Hamiltonian depends on $J_1, J_2, {\cred{}\psi_2}$ only, and is therefore integrable.
The lowest order term in ${\cred{}\psi_2}$ has coefficients proportional to $\sqrt{A_1^q A_2^p}$.
Setting the non-linear non-resonant terms $\mu_{ij}$ equal to zero the Hamiltonian becomes
\[
   H = J_1 + \mu \sqrt{(b J_1 - \sigma q J_2)^q (-a J_1 + p J_2)^p} \cos {\cred{}\psi_2}\,.
\]
Reverting back to the original Euclidean coordinates the resonant term reads
\begin{equation*}
	\Delta H = \Re \left[ ( p_{1} + \imag q_{1} )^{q} ( p_{2} - \sigma \imag q_{2} )^{p} \right] \,,
\end{equation*}
and the Hamiltonian simply is $H = H_2 + \mu \Delta H$.
The {\cred{}functions} $H_2$ and $\Delta H$ are in involution and independent almost everywhere, and thus the system is
Liouville integrable. This is the integrable system we are going to analyze in this paper.
Note that unlike previous work~\cite{NekhoroshevSadovskiiZhilinskii2006,Efstathiou2006,SMPJJ08}, we do not add higher order terms to $\Delta H$ to compactify the system. Moreover, we will study the integrable system defined by the Hamiltonian
$H$, instead of the Hamiltonian $\Delta H$. For the discussion of monodromy the two are
equivalent, since monodromy is a feature of the Liouville foliation, that does not depend on the dynamics
on the individual tori. Also, they share the same singularly reduced system.
However, we are also interested in the physically relevant dynamics of the Birkhoff normal form in the context of KAM theory, and therefore we analyse the Hamiltonian $H$ and not the Hamiltonian $\Delta H$.
In particular, when considering the rotation number of the full system the difference is crucial.
Nevertheless, we will study the momentum map $F = (\Delta H, H_2)$
(which is {\em not} the energy-momentum map for the Hamiltonian $H$),
and the Hamiltonian $H = H_2 + \mu \Delta H$ is a (linear) function of the momenta.

\section{Reduction}
\label{reduction}

In this section we review the steps necessary to reduce to a single degree of freedom using the symmetry $H_2$. We state only the results and refer the interested reader to the standard literature, for example Cushman~\cite{CushmanBates1997}, Abraham and Marsden~\cite{RAbrahamJEMarsden1978} and Broer~et.~al.~\cite{Broer2003}, for the $p:-q$ resonance in particular also see~\cite{Efstathiou2005}. These steps retrace the derivation reviewed in the previous section in a more geometric way.

The flow of the resonant quadratic part $H_2$ is the $S^{1}$--group action
\begin{equation}
\label{H2flow}
\begin{aligned}
\Phi^{H_{2}} : S^{1} \times \C^{2} & \longrightarrow \C^{2} \nonumber\\
\left( t, z_{1}, z_{2} \right)&  \mapsto \left( z_{1} \exp \left( p \imag t \right),
z_{2} \exp \left( \sigma q \imag t \right) \right), \quad z_{1}, \, z_{2} \in \C
\end{aligned}
\end{equation}
where $z_{i} = \frac{1}{\sqrt{2}} \left( p_{i} + \imag q_{i} \right)$, $i=1$, $2$. This action is non--degenerate except on the axis $z_{1}=0$ and $z_{2}=0$ on which points have isotropy subgroup $\Z_{q} \subset S^{1}$ and $\Z_{p} \subset S^{1}$, respectively (if $q,p > 1$). The invariants of this group action are
\begin{equation}
\pi_{1} = z_{1} \bar{z}_1, \quad \pi_{2} = z_{2} \bar{z}_{2}, \quad
\pi_{3} = \Re( z_{1}^{q} z_{2}^{p} ), \quad
\pi_{4} = \Im( z_{1}^{q} z_{2}^{p} )
\end{equation}
for $\sigma=-1$. For $\sigma = +1$ instead
$\pi_{3} = \Re( z_{1}^{q} \bar{z}_2^{p} ), \pi_{4} = \Im( z_{1}^{q} \bar{z}_{2}^{p} )$ are the invariants.

Because of the non-trivial isotropy we employ singular reduction in order to reduce the system to one degree of freedom. The reduced phase space
\begin{equation*}
	P_{h_{2}} = H_{2}^{-1}(h_{2}) / S^{1}
\end{equation*}
is given by the relation
\begin{equation}
\label{reduction:Equ1}
	\pi_{3}^{2} + \pi_{4}^{2} = \pi_{1}^{q} \pi_{2}^{p}, \quad \pi_{1} \ge 0, \;
	\pi_{2} \ge 0,
\end{equation}
hence it is a semi--algebraic variety in the invariants. The ambient Poisson space is endowed with the Poisson structure
\begin{alignat*}{3}
	\left\{ \pi_{1}, \pi_{2} \right\} &= 0, \quad& \left\{ \pi_{1}, \pi_{3} \right\} &= q \mspace{2mu} \pi_{4}\\
	\left\{ \pi_{1}, \pi_{4} \right\} &= - q \mspace{2mu} \pi_{3}, \quad &\left\{ \pi_{2}, \pi_{3} \right\} &= - \sigma  p \mspace{2mu} \pi_{4}\\
	\left\{ \pi_{2}, \pi_{4} \right\} &= p \mspace{2mu} \sigma \pi_{3}, \quad & \left\{ \pi_{3}, \pi_{4} \right\} &= {\textstyle \frac12} {\pi_{1}^{q-1} \pi_{2}^{p-1}}
	\left( \sigma p^{2} \pi_{1} - q^{2} \pi_{2} \right)
\end{alignat*}
with Casimir $\mathcal{H}_{2} = p \pi_{1} + q \sigma \pi_{2}$ and symplectic leaves the reduced phase spaces $P_{h_{2}}$.

Fixing the Casimir $\mathcal{H}_2 = h_{2}$ we choose to eliminate $\pi_{2}$ from~\eqref{reduction:Equ1} and thus get an equation
for the reduced phase space in the form
\begin{equation}
\label{reduction:Equ3}
	\pi_{3}^{2} + \pi_{4}^{2} =
	\pi_{1}^{q} \left(  \frac{ h_{2} - p \pi_{1}}{\sigma q } \right)^{p}.
\end{equation}
The interval of valid $\pi_{1}$ is determined by the requirements $\pi_{1} \ge 0$ and $\pi_{2} \ge 0$, thus
\begin{equation}
\label{intervals}
\begin{alignedat}{3}
	\sigma &= +1: & \pi_{1} &\in [0,\noicefrac{h_{2}}{p}]\\
	\sigma &= -1, \; h_{2} > 0: \; & \pi_{1} &\in [\noicefrac{h_{2}}{p}, \infty)\\
	\sigma &= -1, \; h_{2} \le 0: \; & \pi_{1} &\in [0, \infty).
\end{alignedat}
\end{equation}

\begin{figure}[tbhp]
\includegraphics[width=0.45\textwidth]{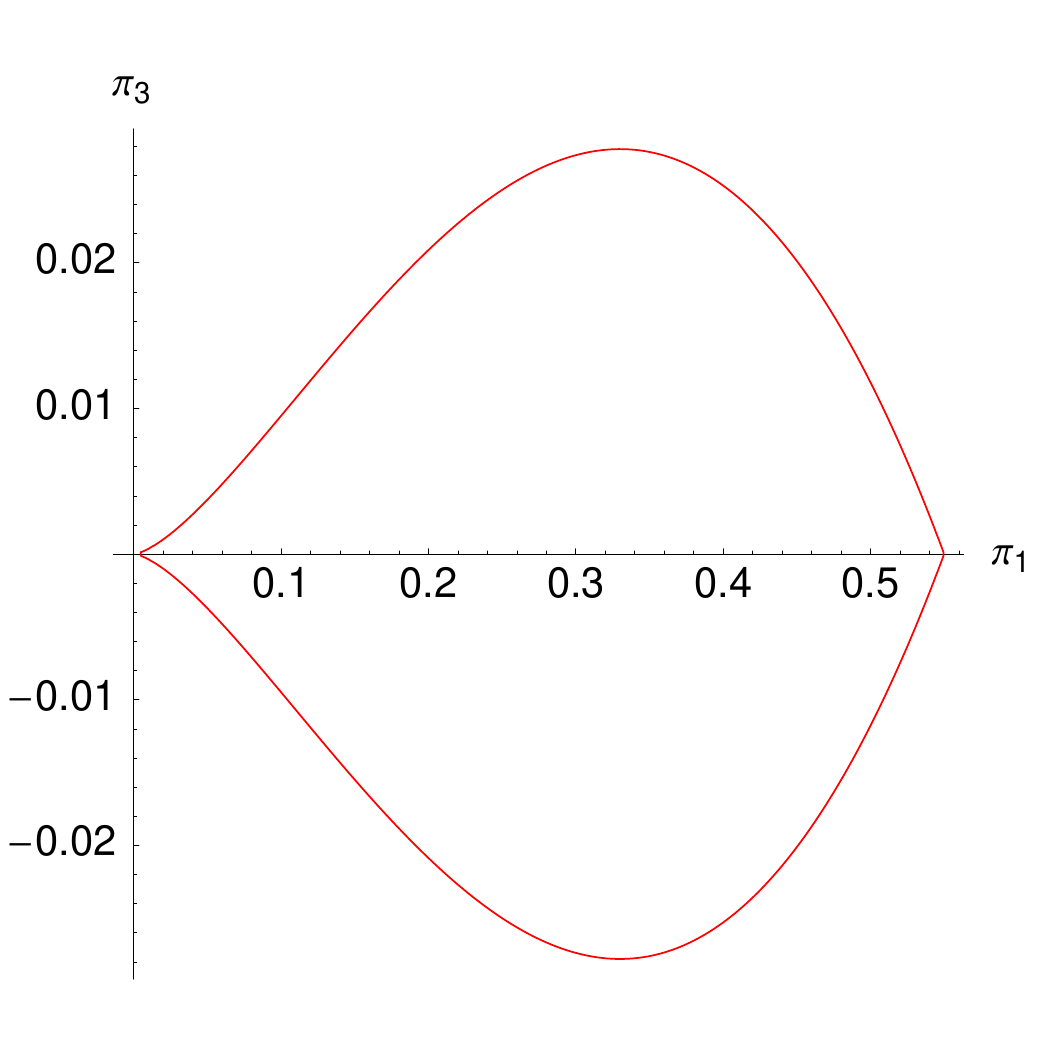}
\includegraphics[width=0.45\textwidth]{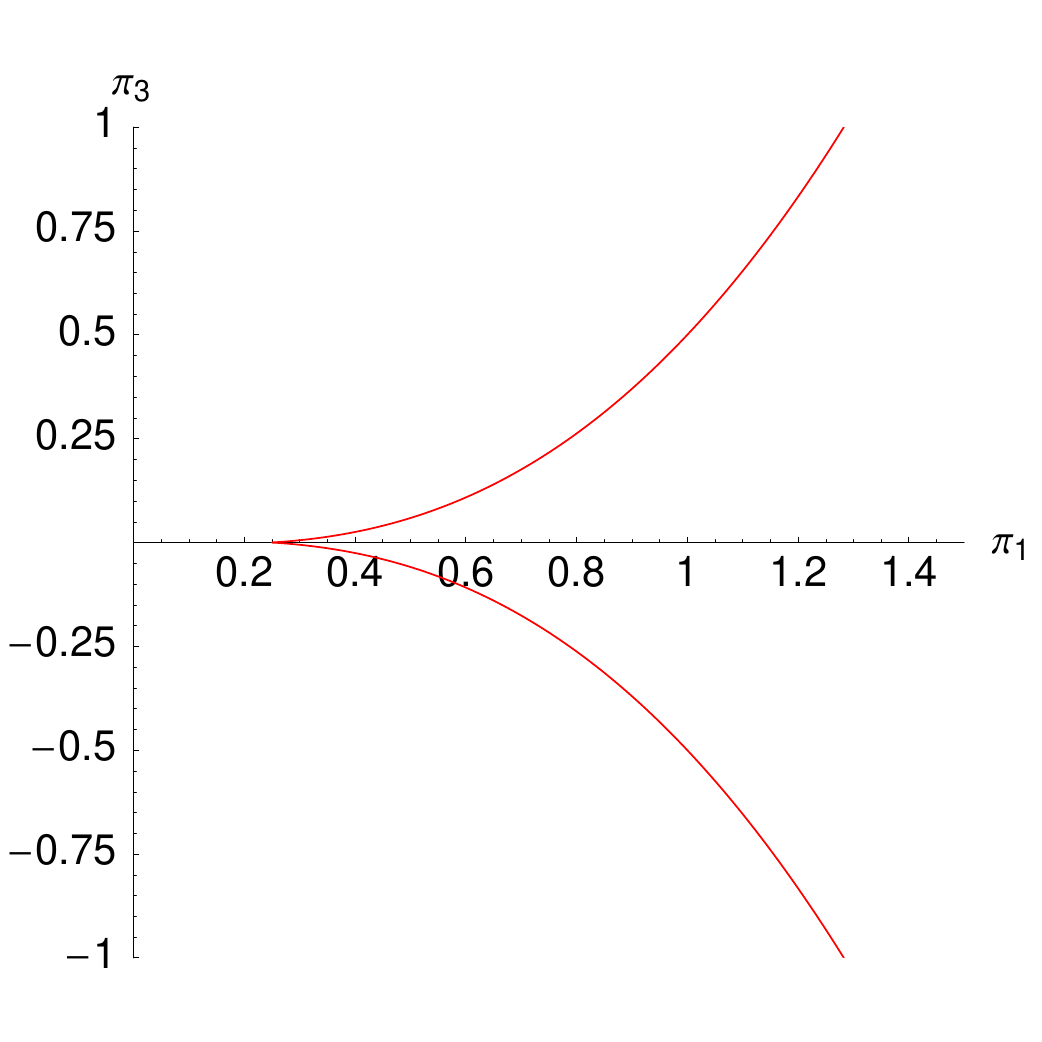} \\
\includegraphics[width=0.45\textwidth]{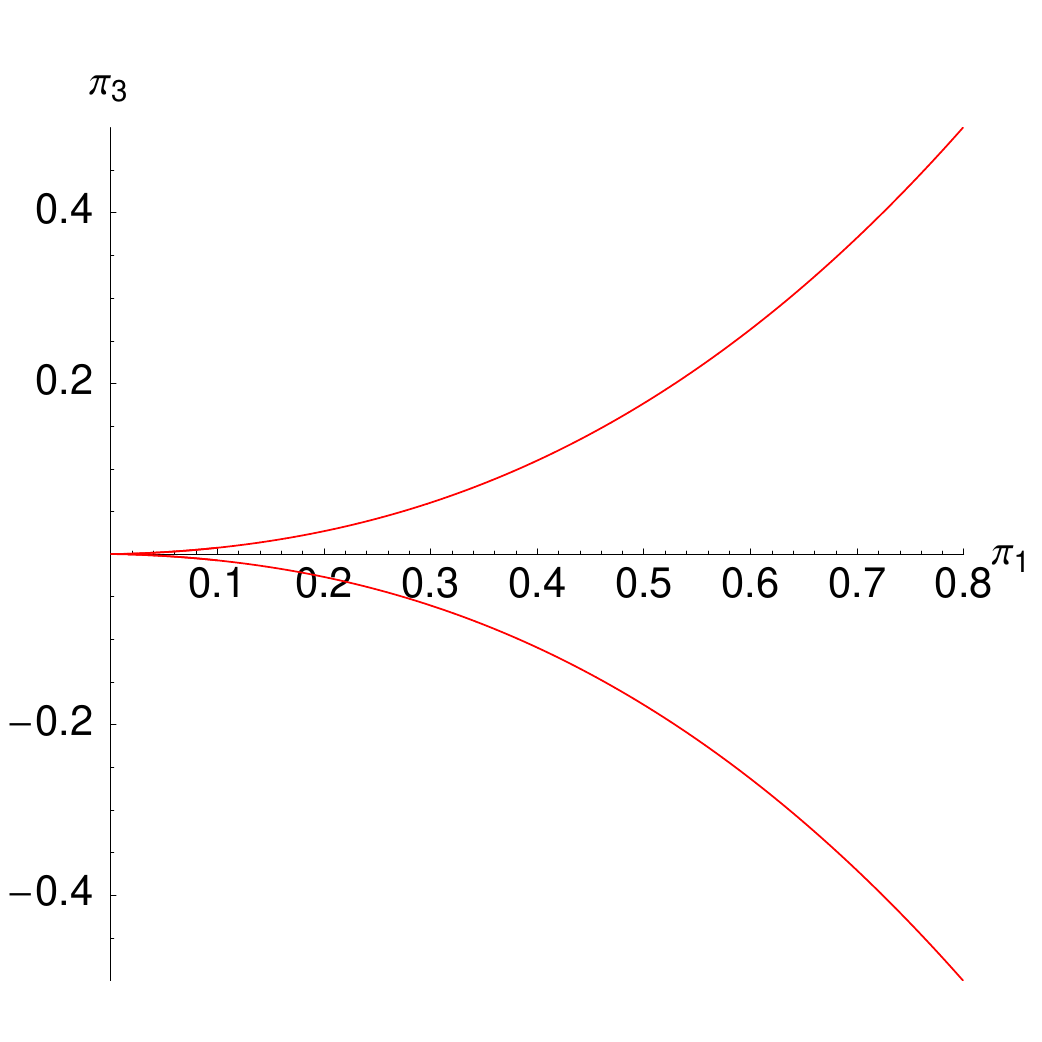}
\includegraphics[width=0.45\textwidth]{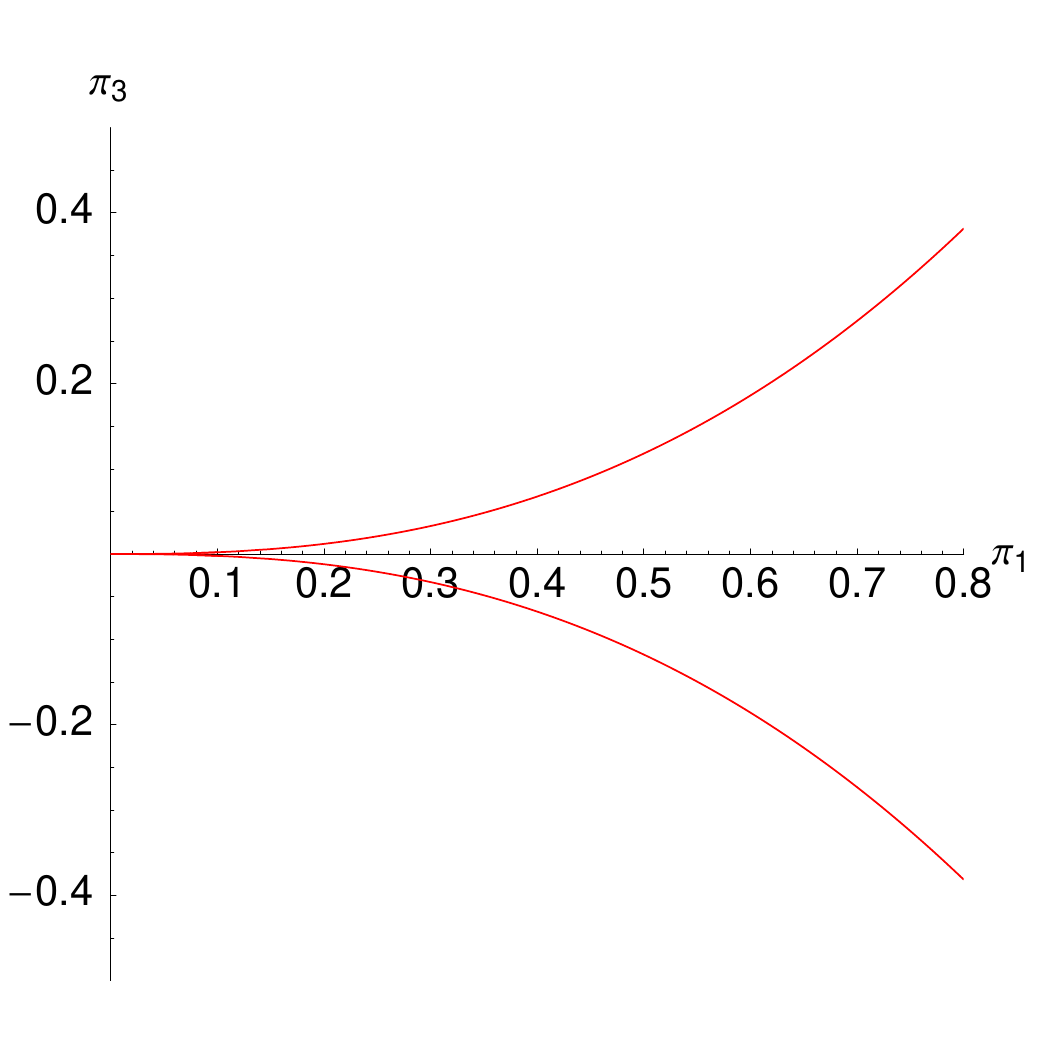}
\caption[The section $\pi_{4}=0$ of the reduced phase space $P_{h_{2}}$ for the $2: \pm 3$ resonance]{\label{reduction:Fig1}The section $\pi_{4}=0$ of the reduced phase space $P_{h_{2}}$ for the $2: \pm 3$ resonance. Upper left: $\sigma=+1$ and $h_2=1.1$. Upper right: $\sigma=-1$, $h_2=0.5$. Lower left: $\sigma=-1$, $h_2=-0.5$. Lower right: $\sigma=-1$, $h_2=0$.
For $h_2 \not = 0$ points $\pi_1 = 0$ are cusp singularities, and points $\pi_1 = h_2/2$ are conical singularities, both due to the non-trivial isotropy of $\Phi^{H_2}$.
 }
\end{figure}

Fig.~\ref{reduction:Fig1} shows sections $\pi_{4}=0$ of the rotationally symmetric reduced phase space $P_{h_{2}}$ for the $2: \pm 3$ resonance for all four relevant cases $\sigma=+1$, $\sigma=-1$ and $h_{2}>0$, $h_{2}<0$ and $h_{2}=0$. For $\sigma=+1$ the reduced phase space is compact with a cusp--singularity ($q=3$) at the origin $\pi_{1}=0$ and a conical singularity ($p=2$) at $\pi_{1} = \noicefrac{h_{2}}{2}$ (where $\pi_{2}=0$). These singularities in the reduced phase space are a result of the non--trivial isotropy of the group action $\Phi^{H_{2}}$ at $z_{1}=\pi_{1}=0$ and $z_{2}=\pi_{2}=0$.

For $\sigma = -1$ the reduced phase space is non-compact. The singular points are of the same type as in the compact case, but
they exist separately for positive or negative $h_2$. For $h_2 > 0$ there is a conical singularity ($p=2$) at $\pi_2 = 0$ and
for $h_2 < 0$ there is a cusp singularity ($q = 3$) at $\pi_1 = 0$. For $h_2$ there is {\cred{} a} singularity of order $p+q= 5$ at $\pi_1 = 0$
(in the compact case $\sigma = 1$ the reduced phase space is merely a point for $h_2 = 0$).

It follows from equation~\eqref{reduction:Equ3} that in the general $p:\pm q$ resonance there is a singularity of order
$q$ at $\pi_1 = 0$ and of order $p$ at $\pi_2 = 0$, assuming $h_2 \not = 0$. If $h_2 = 0$ the order of the singularity is $p+q$.
Note that in the above statement a singularity of order 1 means no singularity.

Expressed in the invariant the integral $\Delta H$ simply becomes $\pi_3$, so that the reduced Hamiltonian is
\begin{equation}
\label{reduction:Equ4}
    \mathcal{H}( \pi_{1}, \pi_{2}, \pi_{3} ) = p \pi_1 + \sigma q \pi_2 + \mu \pi_3 = h_2 + \mu \pi_3  \,.
\end{equation}
As mentioned earlier, the truncated resonant Birkhoff normal form of a generic resonant Hamiltonian system would also contain
terms $\pi_1^i \pi_2^j$, with $i + j \le (p+q)/2$, but in order to maximize the effect of the resonant term these are assumed to be zero.

\section{The Energy--Momentum Mapping}
\label{EMpq}

Let
\begin{equation*}
	F : ({\cred{}p_1}, {\cred{}q_1}, {\cred{}p_2}, {\cred{}q_2}) \mapsto \left( H_{2}({\cred{}p_1}, {\cred{}q_1}, {\cred{}p_2}, {\cred{}q_2}), \Delta H({\cred{}p_1}, {\cred{}q_1}, {\cred{}p_2}, {\cred{}q_2}) \right)
\end{equation*}
be the momentum map of the integrable system, and denoted its value by $(h_{2}, \Delta h)$.
The elliptic equilibrium at the origin is a critical point of the momentum map $F$ since both integrals are of order $\ge 2$ in $(p_{i}, q_{i})$, so that $\rank DF(0) = 0$. Furthermore, this point is {\em degenerate in the sense of the momentum map}~\cite{AVBolsinovATFomenko2004}. For this we need to show that the Hessians of $H_{2}$ and $\Delta H$ are not linear independent at $x=0$. This follows immediately from the fact that $\Delta H$ is of order $p + q \ge 3$ in $(p_{i}, q_{i})$ so that its Hessian vanishes identically at the origin.

\begin{figure}[tbh]
\includegraphics[width=0.45\textwidth]{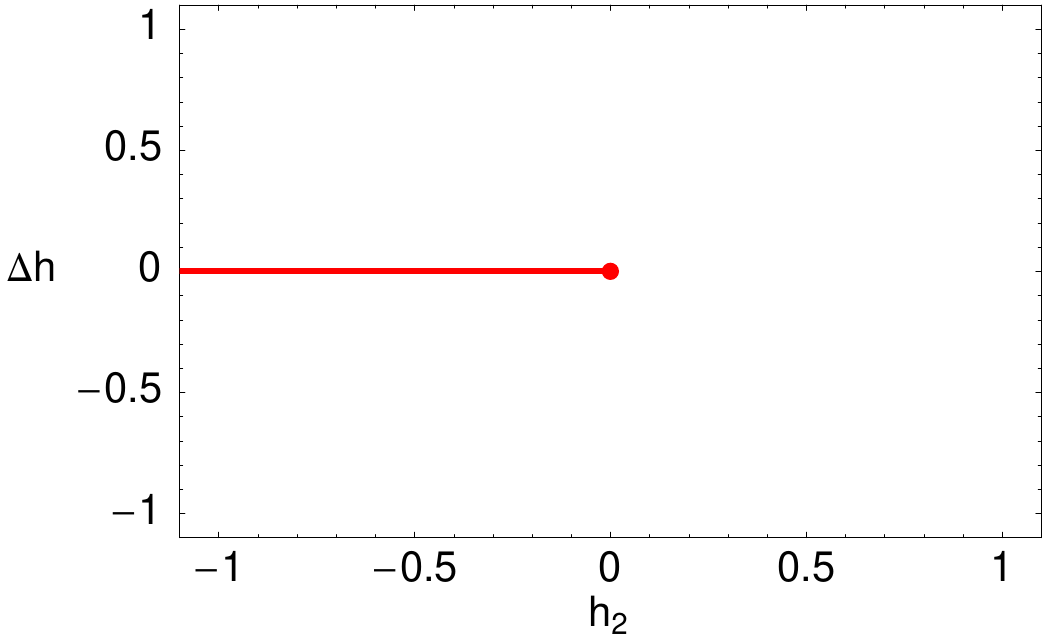}
\includegraphics[width=0.45\textwidth]{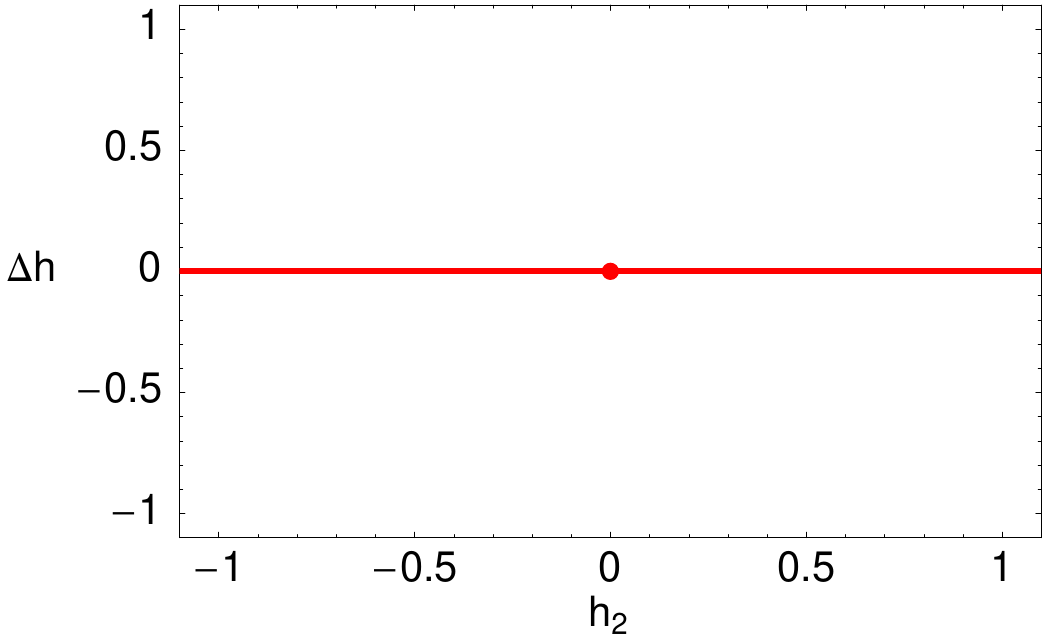} \\
\includegraphics[width=0.45\textwidth]{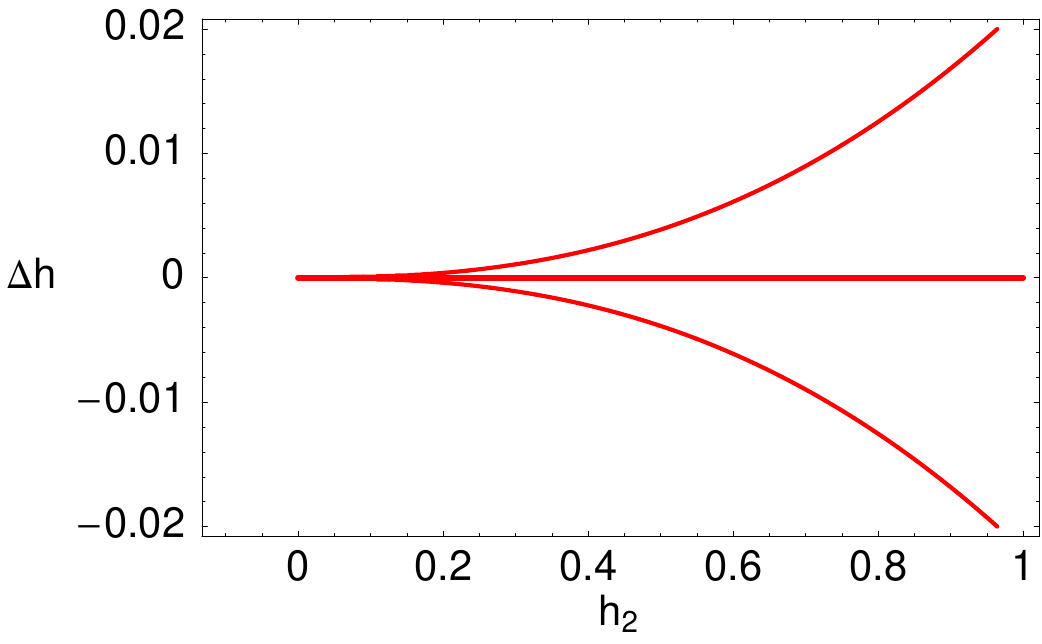}
\caption[The critical values of the energy--momentum mapping]{\label{EMpq:Fig1}The critical values $(h_{2}, \Delta h)$ of the energy--momentum mapping are shown as red branches. The top left picture corresponds to the $1: - q$ resonance, the top right one to the $p : - q$ resonance with $p > 1$. The degenerate equilibrium point is marked by a disk at the origin for $\sigma=-1$.
The lower left picture corresponds to the $2:3$ resonance with $\sigma = +1$.
}
\end{figure}

%

The critical values $(h_{2}, \Delta h)$ of the momentum mapping for the $p: \pm q$ resonance (i.e.~the bifurcation diagram) are shown in fig.~\ref{EMpq:Fig1}, and are described in the following lemma.

\begin{lemma}
The entire line $\Delta h = 0$ is critical for $p > 1$, $q > 2$ and $\sigma=-1$. For $p=1$, the critical values are given by $\Delta h = 0$ and $h_{2} \le 0$ and for $\sigma=+1$ they are given by  $\Delta h = 0$ and $h_{2} \ge 0$. Moreover, there are two additional branches emanating from the origin of the bifurcation diagram for $\sigma=+1$.
\end{lemma}
\begin{proof}
A point in reduced phase space is critical if 1) it is a singular point of the reduced phase space or 2) if it is a point of tangency of the surface $\Delta \mathcal{H} = \Delta h$ (which is simply a horizontal plane) and the reduced phase space. The first condition together with \eqref{intervals} gives that $\pi_1=\pi_3 = \pi_4 = 0$ is critical for $\sigma h_2 \ge 0$ and that $\pi_1 = h_2/p, \pi_3 = \pi_4 = 0$ is critical for $h_2 \ge 0$. The corresponding critical values of the momentum map are $(h_2, 0)$. The second condition requires the gradients of the two integrals to be parallel, thus $\pi_4 = 0$, $\pi_3 \not = 0$,
and the derivative of the right hand side of \eqref{reduction:Equ3}, $\pi_1^q \pi_2^p$, with respect to $\pi_1$ vanishes. This gives $p^2\pi_1 = \sigma q^2 \pi_2$. Since $\pi_1, \pi_2$ are non-negative there is no tangency  for $\sigma = -1$ (appart from the singular point with $h_2 = 0$).
For $\sigma = 1$ there are two additional families of critical values emerging from the origin in a cusp of order $p+q$.
\end{proof}


\section{Dynamics near the Degenerate Equilibrium}

In this section we derive equations for the period $T$ of the reduced flow, the rotation number $W$ of the full system, and the second action $I_{2}$. A prominent feature of the reduced period $T$ is its algebraic (rather than logarithmic) divergence when approaching the degenerate equilibrium point. This is especially easy to see once we introduce weighted polar coordinates $(\rho, \theta)$ in the bifurcation diagram. In these coordinates, $T$ and $W$ separate into $\rho$--~and $\theta$--dependent contributions which considerably simplify the computations.

\subsection{The Reduced Period}
\label{Tpq}

An equation for $T$ is derived by separation of variables from
\begin{equation}
\label{Tpq:Equ1}
	\dot{\pi}_{1} = \left\{ \pi_{1}, \mathcal{H} \right\} = \mu \left\{ \pi_{1}, \pi_{3} \right\} =
	\mu q \pi_{4}.
\end{equation}
Using the equation for the reduced phase space~\eqref{reduction:Equ3} together with $\Delta h = \mu \pi_{3}$, $\pi_4^2$ can be written as a polynomial in $\pi_1$. It follows that the reduced period is defined on the hyperelliptic curve
\begin{equation*}
	\Gamma = \left\{ (\pi_{1}, w) \in \C^{2} \mid w^{2} = Q(\pi_1) \right\} \quad \text{where} \quad
	 Q(z) = \mu^{2} q^{2-p} z^{q} \left( \sigma ( h_{2} - p z ) \right)^{p} - ( q \Delta h )^{2}
\end{equation*}
Separating the variables in equation~\eqref{Tpq:Equ1} and integrating yields
\begin{equation}
	T(h_{2}, \Delta h) = \oint \frac{d \pi_{1}}{w}.
\end{equation}

Our main focus is the $p:-q$ resonance with non--compact fibration. Thus the integral along a closed loop as it stands makes no sense.
We define the reduced period by dividing the dynamics into two parts: Dynamics close to the equilibrium point, and dynamics far away from it.
If the system is compactified by an appropriate higher order term (as in \cite{NekhoroshevSadovskiiZhilinskii2006,Efstathiou2006}) we may assume that the dynamics far away from the equilibrium will eventually return to the neighbourhood of the equilibrium. The time spent on this return loop is a smooth function of initial conditions if we assume that there are no additional critical points in it. Specifically we consider a Poincar\'e section at some small but finite value of $\pi_1$ intersecting stable and unstable invariant
manifolds, similarly to the analysis done for symplectic invariants in \cite{VuNgoc03,DVN05} and the $1:-2$ resonance in \cite{RHCushmanHRDullinHHSS2007}. The contribution of the near-dynamics (from the section of the stable manifold to the section of the unstable manifold) is divergent when approaching the equilibrium point,
while the contribution of the far-dynamics (from the section of the unstable manifold to the section of the stable manifold) remains a smooth and bounded function.

For convenience we modify the truncated period thus defined one more time. Notice that the integral of $d\pi_1/w$
from any finite positive $\pi_1$ to $\infty$ is finite, and smoothly depends on the parameters. Thus changing the truncated period integral to an integral over the non-compact domain only changes it by a smooth function, and the same argument applies. As a result we can treat the closed loop integral of the non-compact system as our leading order period. In particular it correctly describes the leading order divergent terms when approaching the equilibrium point.

An alternative point of view that combines the last two steps (first restriction to the near-dynamics, then the extension to $\infty$)
is to consider the integral in a compactification of the complex plane into a projective space. In this space the integral for $T$ is compact,
and this also explains why it is bounded in the first place. This approach was first used in  \cite{RHCushmanHRDullinHHSS2007}.

In \cite{SMPJJ08} two approaches are followed. The first one generalises the treatment of the
$1:-2$ resonance  \cite{Efstathiou2006,Efstathiou2005} to the $p:-q$ resonance.
There is a privileged compactification which prevents the rotation number from diverging when
approaching the critical values. Using this compactification the period lattice, i.e.~reduced period
and rotation number, is computed. Then {\cred{} the authors} use the Newton polygon to find the leading order
terms of these integrals in the limit approaching the origin of the bifurcation diagram.
The analogue of our polynomial $Q$ appears in their work as the Newton polygon approximation.
What our approach shows is that the resulting hyperelliptic integrals have direct dynamical
meaning as explained above. In the second approach in  \cite{SMPJJ08} the problem is complexified
and monodromy is found as Gauss-Manin monodromy of a loop with complex $(h, \Delta h)$ that avoids critical values.
As {\cred{} the authors} point out the singularity at the origin is not of Morse type, and this is related to the fact the
singularity is degenerate in the sense of the momentum map. As a result the period
diverges algebraically (instead of as usually logarithmically).


%

\begin{lemma}
\label{EPpq:lemma1}
The reduced period $T$ diverges algebraically with exponent $\Abs{\Delta h}^{\frac{2}{p+q}-1}$ upon approaching the degenerate equilibrium point on a curve with non--vanishing derivative at the origin. For $p=1$, the period diverges like
$\Abs{h_{2}}^{-\frac{p+q}{2}+1}$ on the line $\Delta h = 0$ when  $h_{2} \rightarrow 0$.
\end{lemma}
\begin{proof}
Consider the polynomial $Q$ as a polynomial in $\pi_1$, $h_2$, and $\Delta h$.
It is weighted homogeneous, where $h_2$ and $\pi_1$ must have the same weight, so that their weight
is $2$ while that of $\Delta h$ is $p+q$. Therefore we introduce weighted polar coordinates $(\rho, \theta)$
in the image of the momentum map by
\begin{subequations}
\label{Tpq:Equ3}
\begin{align}
	\Delta h &= \rho^{p+q} \sin \theta\\
	   h_{2} &= \rho^{2} \cos \theta
\end{align}
\end{subequations}
Together with $\pi_1 = \rho^2 x$ it follows that
\begin{equation} \label{Teqn}
	T( h_{2}, \Delta h ) = \rho^{-(p+q)+2} \oint \frac{dx}{\tilde{w}}
	=:  \rho^{-(p+q)+2} A(\theta)
\end{equation}
where
\begin{equation*}
	\tilde{w}^{2} = \mu^{2} q^{2-p} x^{q} \left( \sigma( \cos \theta - p x ) \right)^{p}
	- (q \sin \theta)^{2}
\end{equation*}
is independent of $\rho$. Thus $T$ factors into a radial and an angular contribution $A(\theta)$.
The transformation for $(h_2, \Delta h)$ to $(\rho, \theta)$ is $C^{0}$ at the origin but $C^{\infty}$ everywhere else.

The period $T$ diverges with $\rho^{2-p-q}$. When approaching the origin on a line with non-vanishing slope (in the
original variables $(h_2, \Delta h)$ the contribution in $\Delta h \sim \rho^{p+q}$ is of leading order.
Thus $T \sim |\Delta h|^{2/(p+q) - 1}$. When $p=1$ the line segment $\Delta h = 0$ with $h_2>0$ is not critical.
Approaching the origin along this line there is no contribution from $\Delta h$, and thus only $h_2 \sim \rho^2$
is relevant, and the result follows.
\end{proof}

\subsection{The Rotation Number}
\label{Wpq}

Recall from the introduction the symplectic coordinate system with $\{{\cred{}\psi_i}, J_i \} = 1$.
The reduced period gives the period of ${\cred{}\psi_2}$,
while the rotation number gives the advance of ${\cred{}\psi_1}/(2\pi)$ during that period.
The ODE for ${\cred{}\psi_1}$ is
\begin{equation}
\label{Wpq:Equ2}
	\dot{{\cred{}\psi}}_{1} = \left\{ \theta_{1}, H \right\} \,.
\end{equation}
Integration of this ODE for time $T$ gives the rotation number. Note that the angle ${\cred{}\psi_1}$ is
not be globally defined, but all we need is the derivative of ${\cred{}\psi_1}$, see the comments in
\cite{RHCushmanHRDullinHHSS2007}.

The interpretation of the rotation number is similar to the interpretation of the period:
The true (compactified) rotation number will differ from $R$ by a smooth function.
The leading order singular part of the rotation number is contained in $R$.

\begin{proposition}
The rotation number $R$ of the $p : \pm q$ resonance is given by
\begin{equation*}
	R(h_{2}, \Delta h) = \frac{- \sigma}{2 \pi} \oint \left( 1 + \mu \frac{\pi_{3}}{2}
	\left( \frac{q b}{\pi_{1}} - \frac{a p}{\pi_{2}} \right) \right) \frac{d \pi_{1}}{w}.
\end{equation*}
\end{proposition}
\begin{proof}
The angle ${\cred{}\psi_{1}}$ satisfies the following Poisson brackets:
\begin{align*}
	\left\{ {\cred{}\psi_{1}}, \pi_{1} \right\} &= b\\
	\left\{ {\cred{}\psi_{1}}, \pi_{2} \right\} &= -a\\
	\left\{ {\cred{}\psi_{1}}, \pi_{3} \right\} &=
	\frac{\pi_{3}}{2} \left( \frac{q b}{\pi_{1}} - \frac{a p}{\pi_{2}} \right).
\end{align*}
These brackets follow from expressing $\pi_i$ in terms of the canonical variables $(J_i, {\cred{}\psi_i})$, noting $\pi_1 = A_1$, $\pi_2 = A_2$,
$\pi_3 = A_1^{q/2} A_2^{p/2} \cos {\cred{}\psi_2}$, and then computing the canonical brackets, using
\begin{equation*}
	\pfrac{\pi_{1}}{J_{1}} = b, \quad \pfrac{\pi_{2}}{J_{1}} = -a \,.
\end{equation*}

Thus, the differential equation for ${\cred{}\psi_{1}}$~\eqref{Wpq:Equ2} becomes
\begin{equation*}
	\dot{{\cred{}\psi}}_{1} = 1 + \mu \frac{\pi_{3}}{2} \left( \frac{q b}{\pi_{1}} - \frac{a p}{\pi_{2}} \right).
\end{equation*}
Changing the integration variable from $t$ to $\pi_{1}$ using~\eqref{Tpq:Equ1} gives ${\cred{}\psi_1}$ as
an Abelian integral. Comparison of the limiting behaviour of the rotation number and insisting on
the relations $R = \omega_1 / \omega_2$ and $\partial H/\partial I_i = \omega_i$ gives the overall
sign $-\sigma$ in $R$.
\end{proof}

We refer to the three integrals $R$ is composed of as
\begin{align*}
	R(h_{2}, \Delta h) & =  -\sigma \left( \frac{1}{2 \pi} T(h_{2}, \Delta h) + W(h_2, \Delta h) \right), \\
            W(h_2, \Delta h)  & =  W_{1}(h_{2}, \Delta h)	+ W_{2}(h_{2}, \Delta h)
\end{align*}
where
\begin{align*}
	W_{1}(h_{2}, \Delta h) &= \phantom{-} \frac{q b}{4\pi} \Delta h \oint \frac{d \pi_{1}}{\pi_{1} w},\\
	W_{2}(h_{2}, \Delta h) &= -\frac{a p}{4\pi} \Delta h \oint \frac{d \pi_{1}}{\pi_{2} w}.
\end{align*}
Expressing these functions in the weighted polar coordinates $(\rho, \theta)$ gives
\begin{align}
\label{Wpq:Equ1}
	B_{1}(\theta) &= \frac{q b}{4\pi} \sin \theta \oint \frac{d x}{x \tilde{w}}\\
	B_{2}(\theta) &= -\frac{a p q}{4\pi} \sin \theta \oint \frac{d x}{(p x - \cos \theta) \tilde{w}}.
\end{align}
Note that $B_1$ and $B_2$ do not depend on $\rho$, which is the main virtue of the weighted polar coordinates.
Nevertheless, the original functions $W_1$ and $W_2$ are not continuous at $(h_2, \Delta h) = (0,0)$.
The reason is that they take different values when approaching the origin along different lines $\theta = const$.

Figure~\ref{Wpq:Fig1} shows a plot of $B_{1}(\theta)$ on the left and $B_{2}(\theta)$ on the right on a loop $\Gamma(\theta)$ around the origin of the bifurcation diagram fig.~\ref{EMpq:Fig1} top right with constant $\rho$ for the $2:-3$ resonance.

\begin{figure}[tbhp]
\begin{center}
\includegraphics[width=0.45\textwidth]{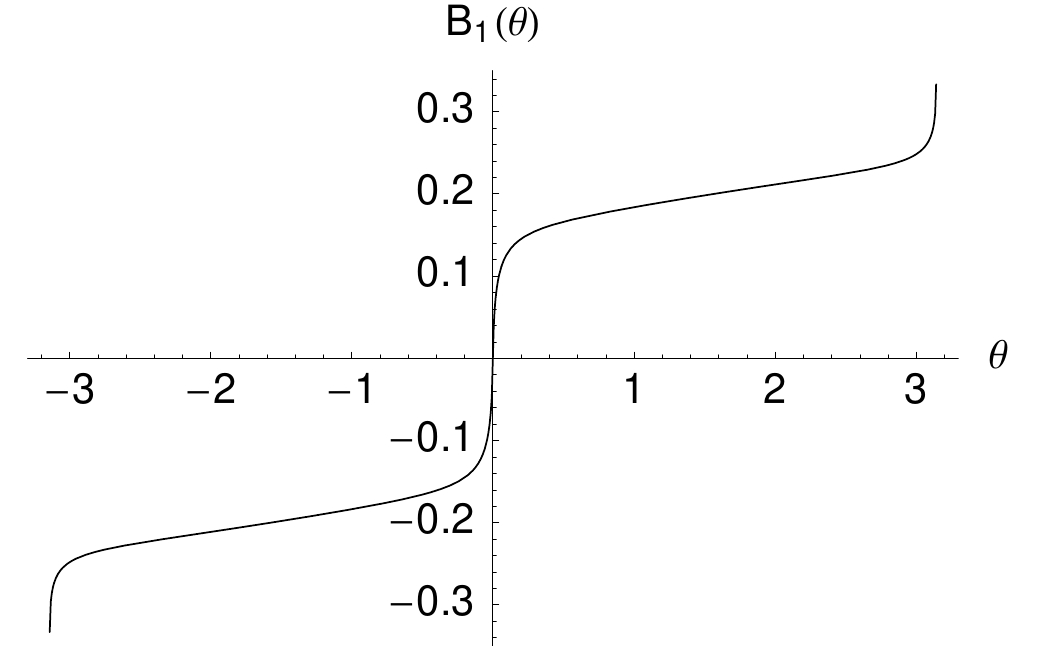}
\includegraphics[width=0.45\textwidth]{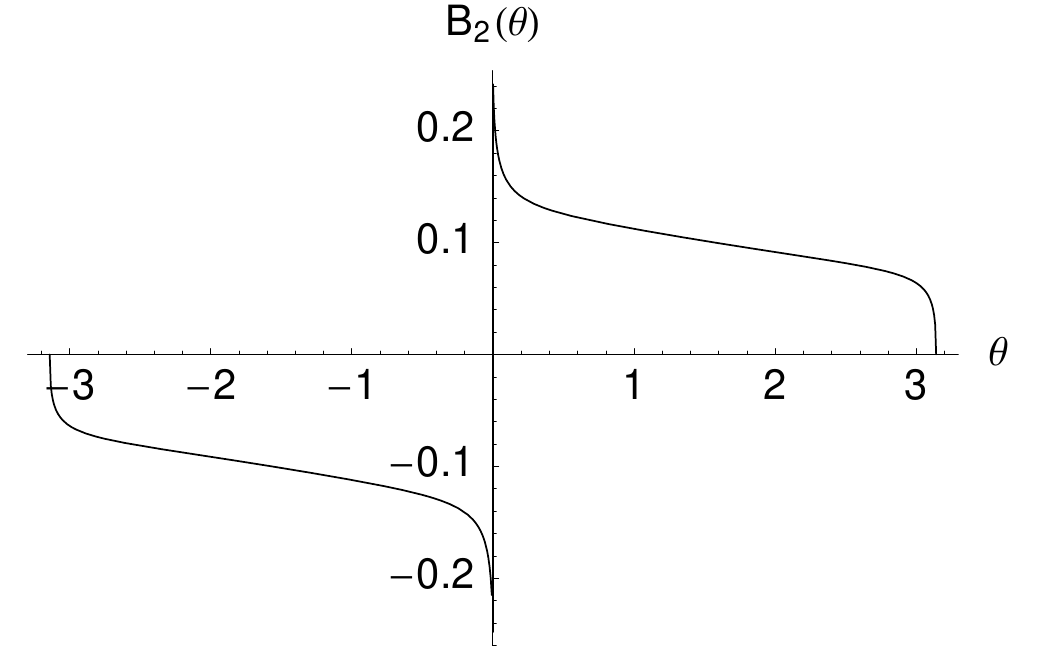}
\end{center}
\caption{\label{Wpq:Fig1}Plot of $B_{1}(\theta)$ (left) and $B_{2}(\theta)$ (right) for the $2:-3$ resonance. $a=-1$, $b=2$.
The corresponding bifurcation diagram is shown in fig.~\ref{EMpq:Fig1} top right.
$B_{1}$ is discontinuous on the branch of critical values with $\Delta h = 0$ and $h_{2} < 0$, i.e.~$\theta = \pm \pi$, with total jump equal to ${\cred{} -b/q}$. $B_{2}$ is discontinuous on the branch of critical values with $\Delta h = 0$ and $h_{2} > 0$, i.e.~$\theta = 0$, with total jump equal to ${\cred{}-a/p}$, see Lemma~\ref{Blem}. }
\end{figure}

The functions $B_{1}$ and $B_{2}$ are periodic, but discontinuous. The discontinuity occurs where the critical points are crossed. The functions can be made continuous by an appropriate shift, but then they will not be periodic any more. They cannot be made periodic and continuous at the same time. It turns out that this behaviour is the reason fractional monodromy occurs in the $p: -q$ resonance, see section~\ref{FMpq}.
From the derivation of $B_1$ and $B_2$ it is clear that they are not really independent functions,
but can be obtained from each other by exchanging $p$ and $q$, and shifting $\theta$ by $\pi$.

Note that if we would have considered $\Delta H$ as our Hamiltonian instead of $H$
the rotation number would not have the diverging contribution from $T$; it would
be $W = W_1 + W_2$ instead of $R = T/2\pi + W$.
For the vanishing twist described later on this contribution is crucial.

\subsection{The Non--Trivial Action}
\label{Apq}

The flow of $H_2$ is periodic with period $2\pi$, and therefore $I_1 = H_2$ is one action of the system.
The other action $I_2$ is a non-trivial function of $\Delta h$ (or $h$) and $I_1$.
From the local canonical coordinates mentioned in the introduction we know that $J_2$ and ${\cred{}\psi_2}$
are conjugate variables. Translating this to the reduced system gives
\begin{proposition}
\label{Apq:prop1}
The non--trivial action $I_{2}$ of the $p : \pm q$ resonance is
\begin{equation}
\label{Apq:Equ1}
	I_{2}( I_{1}, \Delta h) = -\sigma  \left(
	 \frac{p+q}{4\pi} \Delta h \mspace{2mu} T(I_{1}, \Delta h)
	 - I_1 W ( I_1, \Delta h)
	\right).
\end{equation}
\end{proposition}
\begin{proof}
{\cred{}Using} the reduced Poisson bracket we {\cred{}find}
\[
     \{ \cos^{-1} \frac{\pi_3}{\sqrt{ \pi_1^q \pi_2^p }}, a \pi_1 + b \pi_2 \} =-\sigma \,.
\]
The action $I_2$ is therefore given by
\begin{equation*}
	I_{2}(h_{2}, \Delta h) = \frac{\sigma}{2 \pi} \oint \arccos
	\frac{\pi_{3}}{\sqrt{\pi_{1}^{q} \pi_{2}^{p}}} \mspace{2mu}
	( a \mspace{1mu} d \mspace{1mu} \pi_{1} + b \mspace{1mu} d \mspace{1mu} \pi_{2} ).
\end{equation*}
Integration by parts gives 
\begin{equation*}
	I_{2} = \frac{- \sigma}{2 \pi} \oint \frac{\pi_{3}}{\sqrt{\pi_{1}^{q} \pi_{2}^{p} - \pi_{3}^{2}}}
	\left( \left( \frac{a q}{2} + \frac{a p^{2}}{2q} \frac{\pi_{1}}{\pi_{2}} \right) d \pi_{1}
	+ \left( \frac{b p}{2} + \frac{b q^{2}}{2p} \frac{\pi_{2}}{\pi_{1}} \right) d \pi_{2} \right).
\end{equation*}
Now $\pi_1$ and $\pi_2$ are  related by the Casimir $I_1 = h_2 = p \pi_1 + \sigma q \pi_2$.
Thus $d \pi_{2} =  \sigma \frac{p}{q} d \pi_{1}$
and with $\Delta h = \mu \pi_{3}$ this equation becomes
\begin{equation*}
	I_{2} = \frac{- \sigma}{2 \pi} \left( \frac{\Delta h}{2} (p+q) \left( b p - \sigma aq  \right) \oint \frac{d \pi_{1}}{w}
	+ \frac{\Delta h \mspace{2mu} h_{2} \mspace{2mu} a \mspace{2mu} p}{2} \oint
	\frac{d \pi_{1}}{\pi_{2} w}
	- \frac{\Delta h \mspace{2mu} h_{2} \mspace{2mu} b \mspace{2mu} q}{2} \oint
	\frac{d \pi_{1}}{\pi_{1} w} \right).
\end{equation*}
Using $\det M = b p - \sigma a q = 1$, $h_2 = I_1$, and recalling that $W = W_1 + W_2$ gives the result.
Notice that the final answer is independent of the choice of integers $a$ and $b$ in $M$, as long as
$\det M = 1$.
\end{proof}

This expression for $I_{2}$ reduces to the one found for the $1:-2$ resonance \cite{RHCushmanHRDullinHHSS2007}, where $p+q=3$, $a=0$ and $b=1$. In the next section on fractional monodromy we shall prove that the terms $W_{1}$ and $W_{2}$ are discontinuous (see fig.~\ref{Wpq:Fig1}) at $\theta = \pm \pi$ and $\theta=\pm 0$ respectively. They can be made non-periodic and continuous, thus causing the action $I_{2}$ to be globally multivalued.


Even though $T$ diverges algebraically like $\Delta h^{\frac{2}{p+q}-1}$ (Lemma~\ref{EPpq:lemma1}), the action $I_{2}$ goes to zero like $\rho^2 \sim \Delta h^{\frac{2}{p+q}}$ when approaching the equilibrium point.
The action $I_2$ does have the interpretation of a phase space volume. Even though the system is non-compact the action
does not diverge when approaching the equilibrium point. The geometric reason is that $I_2$ measures the volume
relative to the (unbounded) separatrix, which is finite. The boundary terms from the partial integration cancel, see
\cite{Schmidt08} for the details.

Another interpretation of this formula is obtained by solving it for the rotation number.
This gives a decomposition of  the rotation number into a dynamical phase proportional to $T$ and a geometric phase proportional to the action $I_2$, compare e.g.~\cite{Montgomery91}.


\section{Fractional Monodromy}
\label{FMpq}

In this section we establish the fact that the ${\cred{}p: - q}$ resonance has fractional monodromy. We explicitly calculate the monodromy matrix $M$ that gives the transformation of the actions $I_{1}$ and $I_{2}$ after one full anticlockwise cycle around a loop $\Gamma$ enclosing the degenerate equilibrium point at the origin of the bifurcation diagram fig.~\ref{EMpq:Fig1} top row. We then describe the singular fibres corresponding to the critical values the loop $\Gamma$ crosses.

\begin{lemma} \label{Blem}
The function $B_{1}(\theta)$ satisfies
\begin{equation}
	\lim_{\theta \rightarrow \pm \pi^{\mp}} B_{1}(\theta) = \pm  \frac{b}{2q}.
\end{equation}
The function $B_{2}(\theta)$ satisfies
\begin{equation}
	\lim_{\theta \rightarrow 0^{\pm}} B_{2}(\theta) = \mp  \frac{a}{2p}.
\end{equation}
\end{lemma}

This is the main technical result, but instead of duplicating the proof here, we refer to \cite{SMPJJ08}.
The main addition to their work at this point is the interpretation of these integrals. In  \cite{SMPJJ08} they appeared
as the leading order Newton-polygon approximation of the compactified rotation number integrals.
In our approach they appear directly as the integrals of the rotation number of the non-compact system
interpreted as explained before.


Using this technical result we are now giving another proof of fractional monodromy in the $p:-q$ resonance
which is based on the explicit expression of the action obtained earlier. Recall that
$I_{1}$ is the globally smooth action $H_{2}$ of the system as introduced in the introduction.
\begin{theorem}
\label{FMpq:thm1}
The $p:-q$ resonance has fractional monodromy near the degenerate elliptic equilibrium point. The actions change according to
\begin{equation*}
\begin{pmatrix}
	I_{1}'\\[5pt]
	I_{2}'
\end{pmatrix} =
\begin{pmatrix}
	            1 & 0\\[5pt]
	{\cred{}-}\frac{1}{p q} & 1
\end{pmatrix}
\begin{pmatrix}
	I_{1}\\[5pt]
	I_{2}
\end{pmatrix}
\end{equation*}
after one full anticlockwise cycle on a loop $\Gamma$ around the degenerate equilibrium point in the bifurcation diagram.
\end{theorem}
\begin{proof}
Assume the loop $\Gamma$ with fixed $\rho = \rho_{0}$ is traversed in the mathematical positive sense. When crossing the line $\Delta h = 0$, $h_{2} < 0$, $\frac{B_{1}(\pi^{-})}{2 \pi} = \frac{b}{2 q}$, and
$\frac{B_{1}(- \pi^{+})}{2 \pi} = - \frac{b}{2 q}$. Hence, the effective jump of $W_{1}$ becomes ${\cred{}-\frac{b}{q}}$. A similar argument holds for $W_{2}$ at $\Delta h = 0$, $h_{2} > 0$. Although the loop is traversed such that $\theta$ crosses the line $\Delta h = 0$, $h_{2} > 0$ from below, $W_{2}$ has the opposite sign of $W_{1}$. Thus, the effective jump of $W_{2}$ is ${\cred{}-\frac{a}{p}}$. Using the form of the action variable from proposition~\ref{Apq:prop1}, it follows that the action $I_{2}$ changes like
\begin{equation*}
	{\cred{}I_{2} \rightarrow I_{2} - I_{1} \left( \frac{a}{p} + \frac{b}{q} \right)}
\end{equation*}
upon completing a full cycle. Due to
\begin{equation*}
	\frac{a}{p} + \frac{b}{q} = \frac{aq + b p}{p q} = \frac{1}{p q},
\end{equation*}
the monodromy is (independently of the integers $a$ and $b$) given by
\begin{equation*}
	{\cred{}I_{2} \rightarrow I_{2} - \frac{1}{p q} I_{1}}.
\end{equation*}
\end{proof}

The singular fibre of the critical values with $h_2 < 0$ depends on the value of $q$.
Consider the Poincar\'e section $p_2 = 0$ and $\dot p_2 < 0$. The critical level has
$\Delta h = \pi_3 = 0$. Thus we need to study the level set determined by the three equations
$p |z_1|^2 + \sigma q |z_2|^2 = h_2 < 0$, $\Re( z_1^q z_2^p) = 0$, and $p_2 = 0$ near the
critical point $z_1 = 0$. From the first and last equation we obtain
$q_2^2 = (-h_2 + p|z_1|^2)/(-\sigma q)$. For small $|z_1|$ therefore $q_2$, and hence $z_2$, is approximately constant.
The remaining equation thus becomes $\Re( z_1^q) = 0$.
Writing $z_1 = r \exp \imag \varphi$ gives $r_1^q \cos q \varphi = 0$, so that
the level sets of the intersection of the critical fibre with the Poincar\'e section are
given by the $2q$ rays with $q \varphi = \pi ( n + 1/2)$, $n = 0, 1, \dots, 2q - 1$.
Together they form $q$ straight lines which are the stable and unstable manifolds
in alternation, all passing through the origin.

If the level set is compactified, e.g.\ by considering $\Re(z_1^q) + |z_1|^{q+1}$,
it becomes a flower with $q$ petals.
The action of the flow of $H_2$ on the level set is as follows. Start with a point in the section $p_2 = 0$, hence $z_2 = q_2$.
It returns to the section when the imaginary part of its image vanishes again with positive derivative,
hence for $\Im( q_2 \exp( \sigma q i t)) = q_2 \sin (\sigma q t) = 0$, with $\sigma q \, q_2 \cos(\sigma q t)  < 0$.
The smallest positive $t$ that solves this is $t = 2\pi/q$. The action of the flow after the return time therefore
is
\[
     \Phi^{H_2}( 2\pi/q, z_1, z_2) = ( z_1 \exp( 2\imag \pi p / q), z_2 ) \,.
\]
This is simply a rotation by $2\pi p /q$ in the $z_1$ plane.
As a result the petals of the flower are mapped into each other, and $p$ is the number of petals that
the rotation advances by. Since $p$ and $q$ are coprime, all petals are visited before the orbit
returns back to the initial one.
The action is the same for either sign of $\sigma$, the difference is that for $\sigma = -1$ instead of petals we have sectors delineated by
stable and unstable manifolds.

Considering the level set as a whole, and not only in the Poincar\'e section, gives a curled torus whose
transversal cross section is the petal, and whose tubes rotate by $2\pi p/q$ when they complete one
longitudinal cycle. The fact that we discussed the flow of $H_2$ instead of the flow of $H$ (or $\Delta H$)
is immaterial for the topology of the level set since the flows commute.

For $h_2 > 0$ in a similar way a flower with $p$ petals appears in the section $p_1 = 0$.

All critical values are degenerate when $q > 2$. This is a crucial difference between the $1:-2$ resonance
and the higher order resonances. In particular the terms $\mu_{ij}$ that were set to zero in the normal form
could completely change the bifurcation diagram.
How to envisage the singularity at the equilibrium point with $h_2 = \Delta h = 0$ is unclear.
Both curled tori limit to this unstable degenerate equilibrium point,
but the above argument breaks down since $|h_2|$ cannot be assumed large compared to $|z_i|$ any more.

\section{Vanishing Twist in the $1:-q$ Resonance }
\label{VTpq}

For the $p:-q$ resonance with $p>1$ the bifurcation diagram is divided into two halves by the critical line $\Delta h = 0$.
We suspect that in this case the twist does not vanish for regular values, but we have not been able to find a proof of this.
However, when $p=1$ the line of critical values stops at the origin, see fig.~\ref{EMpq:Fig1}. This is therefore a typical
case where one can expect vanishing twist to occur, see \cite{DI04} for a general topological proof.
In this particular case using the weighted polar coordinates allows for a simple analytical proof of vanishing twist.

\begin{theorem}
For the $1:-q$ resonance, $q \ge 2$, the twist vanishes near the degenerate equilibrium point on the curve $\Delta h = 0$, $h_{2} > 0$.
\end{theorem}
\begin{proof}
By definition the twist vanishes if the rotation number $R$ has a critical point on the energy surface, i.e.\ when
$R$ does not change from torus to torus. On the energy surface this is written as $\partial R/\partial I_1|_{H = const}$.
In the image of the energy momentum map the condition is satisfied when the lines of constant rotation number and
the lines of constant energy are tangent to each other. This implies that $\nabla R$ and $\nabla H$ are parallel.
The gradients can be computed in any coordinate system, and we choose the weighted polar coordinates.
Thus the condition for vanishing twist is
\begin{equation*}
	\pfrac{R}{\rho} \pfrac{H}{\theta} - \pfrac{R}{\theta} \pfrac{H}{\rho} = 0 \, .
\end{equation*}
Now $H = H_2 + \mu \Delta H$ where $H_2$ is order 2 in $\rho$ and $\Delta H$ is order $p+q$ in $\rho$.
For small $\rho$ we can therefore neglect $\Delta H$, and find
$\partial H/ \partial \rho \approx  2 \rho \cos \theta$, $\partial H/\partial \theta \approx  -\rho^2 \sin \theta$.
Recalling the factorisation  $T = \rho^{-(p+q)+2}A(\theta)$ from \eqref{Teqn} we find
\begin{align*}
	  - \sigma 2 \pi \pfrac{R}{\rho} &= \pfrac{T}{\rho} = \left( -(p+q)+2 \right)
	  \rho^{-(p+q)+1} A(\theta) , \\
	- \sigma 2 \pi \pfrac{R}{\theta} &= \pfrac{T}{\theta} + {B_{1}'}(\theta) +{B_{2}'}(\theta) \approx \rho^{-(p+q)+2} A'(\theta) \,.
\end{align*}
The leading order contribution comes from $T$ alone. Altogether this gives
\begin{equation*}
\label{VTpq:Equ1}
	- \sigma 2 \pi \left( \pfrac{R}{\rho} \pfrac{H}{\theta} - \pfrac{R}{\theta} \pfrac{H}{\rho} \right) =
	\rho^{-(p+q)+3} \left( (p+q-2) A(\theta) \sin \theta - 2 A'(\theta) \cos \theta
	+ O(\rho)\right)
\end{equation*}
up to lowest order in $\rho$. The solution $\theta = 0$ follows from symmetry:
$A(\theta)$ is even in $\theta$ thus $A'(\theta)$ is odd in $\theta$, and thus both terms are odd in $\theta$.
For $p = 1$ the line $\theta = 0$ is a line of regular values, and therefore we have shown that the twist
vanishes along $\theta = 0$ asymptotically near the origin.
\end{proof}


\subsection*{Acknowledgements}

The authors would like to thank the Department of Applied Mathematics at the University of Colorado, Boulder for its hospitality in 2006/7 where this work was completed.
HRD was supported in part by a Leverhulme Research Fellowship.

\bibliographystyle{plain}
\bibliography{pqResonance}

\end{document}